\documentclass{lmcs} 
\pdfoutput=1
\usepackage[utf8]{inputenc}

\usepackage{lastpage}
\lmcsdoi{22}{2}{10}
\lmcsheading{}{\pageref{LastPage}}{}{}%
{Aug.~13,~2024}{Apr.~29,~2026}{}


\usepackage{hyperref}
\usepackage{cleveref}
\usepackage{xspace}
\usepackage[shortcuts]{extdash}
\usepackage{mathtools}

\newcommand{\re}[1]{\xrightarrow{#1}}
\newcommand{\rp}[1]{\overset{#1}{\rightsquigarrow}}

\newcommand{\tand}{\text{ and }}

\newcommand{\tin}{\text{ in }}

\newcommand{\Parity}{\text{Parity}}

\newcommand{\eps}{\varepsilon}

\newcommand{\VE}{V_{\text{Eve}}}
\newcommand{\VA}{V_{\text{Adam}}}

\newcommand{\N}{\mathbb N}
\newcommand{\Z}{\mathbb Z}

\newcommand{\fin}[1]{W_{\mathrm{fin}}}

\newcommand{\MP}{\text{Mean-Payoff}}
\renewcommand{\mp}{\text{mp}}
\newcommand{\Bounded}{\text{Bounded}}

\newcommand{\normal}{\mathcal{N}\xspace}
\newcommand{\cobuchi}{\mathcal{F}\xspace}

\newcommand{\coBuchi}{\text{co-Büchi}}

\newcommand{\reb}[1]{%
  \mathrel{\ooalign{\hfil$\vcenter{
    \hbox{$\scriptscriptstyle\bullet$}}$\hfil\cr$\xrightarrow{#1}$\cr}
  }%
}

\renewcommand{\ss}{\mathrm{s}}

\newcommand{\ENI}{\text{ENI}}
\newcommand{\END}{\text{END}}
\newcommand{\Finite}{\text{Finite}}

\newcommand{\TB}[1]{\text{Tilted-Bounded}_{#1}}

\newcommand{\ident}[1]{\ensuremath{\texorpdfstring{\mathrm{#1}}{#1}}\xspace}

\newcommand{\lang}{\ident{L}}

\newcommand{\boldclass}[3]{\texorpdfstring{\ensuremath{\mathbf{#1}^{#2}_{#3}}}{Borel}}

\newcommand{\bsigma}[1]{\boldclass{\Sigma}{0}{#1}}

\newcommand{\bdelta}[1]{\boldclass{\Delta}{0}{#1}}

\theoremstyle{plain}

\newtheorem{theorem}[thm]{Theorem}
\newtheorem{lemma}[thm]{Lemma}
\newtheorem{conjecture}[thm]{Conjecture}
\newtheorem{claim}[thm]{Claim}
\newtheorem{proposition}[thm]{Proposition}

\crefname{conjecture}{Conjecture}{Conjectures}

\begin{document}

\title{Positionality in \texorpdfstring{$\bsigma{2}$}{Sigma02} and a completeness result}
\thanks{We thank Antonio Casares and Lorenzo Clemente for discussions on and around the topic.\\
The first author was supported by the European Research Council (grant agreement No 948057 — BOBR).\\
The second author was supported by the National Science Centre, Poland (grant no.\@ 2021/\allowbreak41/\allowbreak B/\allowbreak ST6/\allowbreak03914).}

\author[P.~Ohlmann]{Pierre Ohlmann\lmcsorcid{0000-0002-4685-5253}}[a]
\author[M.~Skrzypczak]{Michał Skrzypczak\lmcsorcid{0000-0002-9647-4993}}[b]

\address{CNRS, Aix-Marseille Université, Laboratoire d'Informatique et des Systèmes (LIS), France}	
\email{pierre.ohlmann@lis-lab.fr}  

\address{Institute of Informatics, University of Warsaw, Poland}	
\email{mskrzypczak@mimuw.edu.pl}  


\begin{abstract}
We study the existence of positional strategies for the protagonist in infinite duration games over arbitrary game graphs.
We prove that prefix-independent objectives in $\bsigma{2}$ which are positional and admit a (strongly) neutral letter are exactly those that are recognised by history-deterministic monotone co-Büchi automata over countable ordinals.
This generalises a criterion proposed by [Kopczy\'nski, ICALP 2006] and gives an alternative proof of closure under union for these objectives, which was known from [Ohlmann, TheoretiCS 2023].

We then give two applications of our result.
First, we prove that the mean-payoff objective is positional over arbitrary game graphs.
Second, we establish the following completeness result: for any objective $W$ which is prefix-independent, admits a (weakly) neutral letter, and is positional over finite game graphs, there is an objective $W'$ which is equivalent to $W$ over finite game graphs and positional over arbitrary game graphs.
\end{abstract}

\maketitle

\section{Introduction}

\subsection{Context}

Before presenting our contributions, we introduce the necessary background.

\paragraph*{Games.}
We study infinite duration games on graphs. In such a game, two players, Eve and Adam, alternate forever in moving a token along the edges of a directed, possibly infinite graph (called \emph{arena}), whose edges are labelled with elements of some set $C$.
An \emph{objective} $W \subseteq C^\omega$ is specified in advance; Eve wins the game if the label of the produced infinite path belongs to $W$.
A \emph{strategy} in such a game is called \emph{positional} if it depends only on~the~current vertex occupied by the token, regardless of the history of the play.

We are interested in \emph{positional objectives}: those for which existence of a winning strategy for Eve entails existence of a winning positional strategy for Eve, on a arbitrary arena. Sometimes we also consider a weaker property: an objective is \emph{positional over finite arenas} if the above implication holds on any finite arena.

\paragraph*{Early results.}
Although the notion of positionality is already present in Shapley's seminal work~\cite{Shapley53}, the first positionality result for infinite duration games was established by Ehrenfeucht and Mycielsky~\cite{EM79}, and it concerns the mean-payoff objective
\[
    \MP_{\leq 0} = \Big\{w_0 w_1 \dots \in \Z^\omega \mid \limsup_k \frac 1 k \sum_{i=0}^{k-1} w_i \leq 0\Big\},
\]
over finite arenas.
Nowadays, many proofs are known that establish positionality of mean\=/payoff games over finite arenas.

Later, and in a different context, Emerson and Jutla~\cite{EJ91} as well as Mostowski~\cite{Mostowski:1991} independently established positionality of the parity objective
\[
    \Parity_d = \Big\{p_0 p_1 \dots \in \{0,1,\dots,d\}^\omega \mid \limsup_k p_k \text{ is even}\Big\}
\]
over arbitrary arenas.
This result was used to give a direct proof of the possibility of complementing automata over infinite trees, which is the key step in modern proofs of Rabin's theorem on decidability of S2S~\cite{Rabin69}.
By now, several proofs are known for positionality of parity games, some of which apply to arbitrary arenas.

Both parity games and mean\=/payoff games have been the object of considerable attention over the past three decades; we refer to~\cite{NathBook} for a thorough exposition.
By symmetry, these games are positional not only for Eve but also for the opponent, a property we call \emph{bi-positionality}.
Parity and mean\=/payoff objectives, as well as the vast majority of objectives that are considered in this context, are \emph{prefix\=/independent}, that is, invariant under adding or removing finite prefixes.

\paragraph*{Bi-positionality.}
Many efforts were devoted to understanding positionality in the early 2000's.
These culminated in Gimbert and Zielonka's work~\cite{GZ05} establishing a general characterisation of bi\=/positional objectives over finite arenas, from which it follows that an objective is bi\=/positional over finite arenas if and only if it is the case for 1\=/player games.
On the other hand, Colcombet and Niwi\'nski~\cite{CN06} established that bi\=/positionality over arbitrary arenas is very restrictive: any prefix\=/independent objective which is bi-positional over arbitrary arenas can be recast as a parity objective.

Together, these two results give a good understanding of bi\=/positional objectives, both over finite and arbitrary arenas.

\paragraph*{Positionality for Eve.}
In contrast, less is known about those objectives which are positional for Eve, regardless of the opponent (this is sometimes called half\=/positionality).
This is somewhat surprising, considering that positionality is more in\=/line with the primary application in synthesis of reactive systems, where the opponent, who models an antagonistic environment, need not have structured strategies.
The thesis of Kopczy\'nski~\cite{Kopczynskithesis} proposes a number of results on positionality, but no characterisation.
Kopczy\'nski proposed two classes of prefix\=/independent objectives, \emph{concave objectives} and \emph{monotone objectives}, which are positional respectively over finite and over arbitrary arenas.
Both classes are closed under unions, which motivated the following conjecture.

\begin{conjecture}[Kopczy\'nski's conjecture~\cite{Kopczynskithesis, Kopczynski06}]\label{conj:kop}
Prefix\=/independent positional objectives are closed under unions.
\end{conjecture}

This conjecture was disproved by Kozachinskiy in the case of finite arenas~\cite{Kozachinskiy22}, however, it remains open for arbitrary ones (even in the case of countable unions instead of unions).

\paragraph*{Neutral letters.}

Many of the considered objectives contain a \emph{neutral letter}, that is an~element $\eps\in C$ such that $W$ is invariant under removing arbitrary many occurrences of the letter $\eps$ from any infinite word. For instance, $\eps=0$ is a neutral letter of the parity objective $\Parity_d$. There are two variants of this definition, \emph{strongly neutral letter} and \emph{weakly neutral letter}, which are formally introduced in the preliminaries.
It an open question whether adding a neutral letter to a given objective may affect its positionality~\cite{Kopczynskithesis,Ohlmann23}, although Casares and Ohlmann established that this is not the case for $\omega$-regular objectives~\cite{CO24}.

Neutral letters are typically used when one wants to modify a~given game arena, by allowing players to make some additional decisions.
This requires to create intermediate edges in such a~way that their labels do not affect the overall outcome of the play.
For this reason, a number of results (including those in this paper) are conditional on the presence of a (strong or weak) neutral letter.
For more discussion on neutral letters and their role in the study of memory, see also~\cite{CO25}.

\paragraph*{Borel classes.}

To stratify the complexity of the considered objectives we use the Borel hierarchy~\cite{kechris_descriptive}. This follows the classical approach to Gale\=/Stewart games~\cite{gale_games}, where the determinacy theorem was gradually proved for more and more complex Borel classes: $\bsigma{2}$ in~\cite{wolfe_determinacy} and $\bsigma{3}$ in~\cite{davis_infinite_games}. This finally led to Martin's celebrated result on all Borel objectives~\cite{martin_borel_determinacy}.

To apply this technique, we assume for the rest of the paper that $C$ is at most countable. Thus, $C^\omega$ is a Polish topological space, with open sets of the form $L\cdot C^\omega$ where $L\subseteq C^\ast$ is arbitrary. Closed sets are those whose complement is open. The class $\bsigma{2}$ contains all sets which can be obtained as a~countable union of some closed sets.

\paragraph*{Recent developments.}
A step forward in the study of positionality (for Eve) was recently made by Ohlmann~\cite{Ohlmann23} who established that an objective admitting a (strongly) neutral letter is positional over arbitrary arenas if and only if it admits well\=/ordered monotone universal graphs.
Note that this characterisation concerns only positionality over arbitrary arenas.
This allowed Ohlmann to prove closure of prefix\=/independent positional objectives (over arbitrary arenas) admitting a (strongly) neutral letter under finite lexicographic products, and, further assuming membership in $\bsigma{2}$, under finite unions\footnote{In~\cite{Ohlmann23}, an assumption called ``non-healing'' is used. This assumption is in fact implied by membership in $\bsigma{2}$.}.

Bouyer, Casares, Randour, and Vandenhove~\cite{BCRV22} also used universal graphs to characterise positionality for objectives recognised by deterministic B\"uchi automata. They observed that for such an objective $W$ finiteness of the arena does not impact positionality: $W$ is positional over arbitrary arenas if and only if it is positional over finite ones.

Going further, Casares~\cite{CasaresThesis} recently proposed a characterisation of positionality for all $\omega$\=/regular objectives.
As a by\=/product, it follows that Conjecture~\ref{conj:kop} holds for $\omega$\=/regular objectives\footnote{In fact, Casares proved a strengthening of the conjecture when only one objective is required to be prefix-independent.}, and that again finiteness of the arena does not impact positionality.

\subsection{Contributions}

We now present our contributions.

\paragraph*{Positionality in $\bsigma{2}$.}

As mentioned above, Kopczy\'nski introduced the class of \emph{monotonic objectives}, defined as those of the form $C^\omega \setminus L^\omega$, where $L$ is a language recognised by a finite linearly\=/ordered automaton with certain monotonicity properties on transitions. He then proved that monotonic objectives are positional over arbitrary arenas.
Such objectives are prefix\=/independent and belong to $\bsigma{2}$; our first contribution is to extend Kopczy\'nski's result to a complete characterisation (up to neutral letters) of positional objectives in $\bsigma{2}$.

\begin{restatable}[]{theorem}{charac}\label{thm:charac-sigma2}
Let $W \subseteq C^\omega$ be a prefix\=/independent $\bsigma{2}$ objective admitting a strongly neutral letter.
Then $W$ is positional over arbitrary arenas if and only if it is recognised by a countable history\=/deterministic well\=/founded monotone co\=/Büchi automaton.


\end{restatable}

The proof of Theorem~\ref{thm:charac-sigma2} is based on Ohlmann's \emph{structuration} technique which is the key ingredient to the proof of~\cite{Ohlmann23}.
As an easy by\=/product of the above characterisation, we reobtain the result that Kopczynski's conjecture holds for countable unions of $\bsigma{2}$ objectives (assuming that the given objectives all have strongly neutral letters).

\begin{restatable}[]{corollary}{closureunion}\label{cor:closure-union}%
If $W_0, W_1, \ldots$ are all positional prefix\=/independent $\bsigma{2}$ objectives, each admitting a strongly neutral letter (possibly a different one), then the union $\bigcup_{i\in\N} W_i$ is also positional.
\end{restatable}

\paragraph*{From finite to arbitrary arenas.}


The most important natural example of an objective which is positional over finite arenas but not on infinite ones is $\MP_{\leq 0}$, as defined above.
As a straightforward consequence of its positionality~\cite[Theorem~3]{BCDGR11}, it holds that over finite arenas, $\MP_{\leq 0}$ coincides with the energy condition
\[
    \Bounded=\Big\{w_0 w_1 \dots \in \Z^\omega \mid \sup_{k}\sum_{i=0}^{k-1} w_i \text{ is finite}\Big\},
\]
which turns out to be positional even over arbitrary arenas~\cite{Ohlmann23}.

Applying Corollary~\ref{cor:closure-union}, we establish that with strict threshold, the mean-payoff objective
\[
    \MP_{<0} = \Big\{w_0 w_1 \dots \in \Z^\omega \mid \limsup_k \frac 1 k \sum_{i=0}^{k-1} w_i < 0\Big\}
\]
is in fact positional over arbitrary arenas.

Now say that two prefix\=/independent objectives are \emph{finitely equivalent}, written $W \equiv W'$, if they are won by Eve over the same finite arenas.
As observed above, $\MP_{\leq 0} \equiv \Bounded$, which is positional over arbitrary arenas.
Likewise, its complement
\[
    \Z^\omega \setminus \MP_{\leq 0} = \Big\{w_0w_1 \dots \in \Z^\omega \mid \limsup_k \frac 1 k \sum_{i=0}^{k-1} w_i > 0 \Big\}
\] 
is, up to changing each weight $w\in \Z$ by the opposite one ${-}w\in\Z$, isomorphic to
\[
    \Big\{w_0w_1 \dots \in \Z^\omega \mid \liminf_k \frac 1 k \sum_{i=0}^{k-1} w_i < 0 \Big\}.
\]
The latter condition is finitely equivalent to $\MP_{<0}$ (where the liminf is replaced with a limsup), which, as explained above, turns out to be positional over arbitrary arenas.

Thus, both $\MP_{\leq 0}$ and its complement are finitely equivalent to objectives that are positional over arbitrary arenas.
This brings us to our main contribution, which generalises the above observation to any prefix\=/independent objective admitting a (weakly) neutral letter which is positional over finite arenas.

\begin{restatable}[]{theorem}{main}\label{thm:main}
Let $W \subseteq C^\omega$ be a prefix\=/independent objective which is positional over finite arenas and admits a weakly neutral letter.
Then there exists a prefix\=/independent objective $W'\subseteq W$ such that $W' \equiv W$ and $W'$ is positional over arbitrary arenas.
\end{restatable}


\paragraph*{Structure of the paper} Section~\ref{sec:preliminaries} introduces all necessary notions, including Ohlmann's structurations results.
Section~\ref{sec:sigma2} proves our characterisation result Theorem~\ref{thm:charac-sigma2} and its consequence Corollary~\ref{cor:closure-union}, and provides a few examples.
Then we proceed in Section~\ref{sec:finite-to-infinite} with establishing positionality of $\MP_{<0}$ over arbitrary arenas, and proving Theorem~\ref{thm:main}.

\section{Preliminaries}\label{sec:preliminaries}

\paragraph*{Graphs.}
We fix a set of letters $C$, which we assume to be at most countable.
A \emph{$C$\=/graph} $G$ is comprised of a (potentially infinite) set of \emph{vertices} $V(G)$ together with a set of \emph{edges} $E(G) \subseteq V(G) \times C \times V(G)$.
An edge $e=(v,c,v')\in E(G)$ is written $v \re c v'$, with $c$ being the \emph{label} of this edge.
We say that $e$ is \emph{outgoing} from $v$, that it is \emph{incoming} to $v'$, and that it is \emph{adjacent} to both $v$ and to $v'$.
We assume that each vertex $v\in V(G)$ has at least one outgoing edge (we call this condition being \emph{sinkless}, with a~\emph{sink} understood as a~vertex where no outgoing edge is available).

We say that $G$ is \emph{finite} (resp.~\emph{countable}) if both $V(G)$ and $E(G)$ are finite (resp.~countable). The \emph{size} of a graph is defined to be $|G|=|V(G)|$. We write $\N$ and $\omega$ for the set of natural numbers $\{0,1,\ldots\}$. The symbol $\omega$ is used to emphasise the order\=/type of this set, while $\N$ emphasises the arithmetic structure. $\Z$ stands for the set of all integers $\{\ldots,-1,0,1,\ldots\}$.

A (finite) \emph{path} is a (finite) sequence of edges with matching endpoints, meaning of the form $v_0 \re{c_0} v_1,v_1 \re{c_1} v_2, \dots$, which we conveniently write as $v_0 \re{c_0} v_1 \re{c_1} \dots$.
We say that $\pi$ is a \emph{path from $v_0$ in $G$}, and that vertices $v_0,v_1,v_2,\dots$ appearing on the path are \emph{reachable} from $v_0$.
We use $G[v_0]$ to denote the restriction of $G$ to vertices reachable from $v_0$.
The \emph{label} of a path $\pi$ is the sequence $c_0c_1 \dots$ of labels of its edges; it belongs to $C^\omega$ if $\pi$ is infinite and to $C^\ast$ otherwise.
We sometimes write $v \rp w$ to say that $w$ labels an infinite path from $v$, or $v \rp w v'$ to say that $w$ labels a finite path from $v$ to $v'$.
We write $\lang(G,v_0) \subseteq C^\omega$ for the set of labels of all infinite paths from $v_0$ in $G$, and $\lang(G) \subseteq C^\omega$ for the set of labels of all infinite paths in $G$, that is the union of $\lang(G,v_0)$ over all $v_0\in V(G)$.

A \emph{graph morphism} from $G$ to $G'$ is a map $\phi\colon V(G) \to V(G')$ such that for every edge $v \re c v' \in E(G)$, it holds that $\phi(v) \re c \phi(v') \in E(G')$.
We write $G \re{\phi} G'$.
We sometimes say that $G$ \emph{embeds} in $G'$ or that $G'$ \emph{embeds} $G$, and we write $G \to G'$, to say that there exists a~morphism from $G$ to $G'$.
Note that $G \to G'$ implies $\lang(G) \subseteq \lang(G')$.

A graph $G$ is $v_0$\=/\emph{rooted} if it has a distinguished vertex $v_0\in V(G)$ called the \emph{root}.
A \emph{tree}~$T$ is a $t_0$\=/rooted graph such that all vertices in $T$ admit a unique finite path from the root $t_0$.

\paragraph*{Games.}
A \emph{$C$\=/arena} is given by a $C$\=/graph $A$ together with a~partition of its vertices $V(A) = \VE \sqcup \VA$ into those controlled by Eve $\VE$ and those controlled by Adam $\VA$.
A \emph{strategy} (for Eve) $(S,\pi)$ in an arena $A$ is a graph $S$ together with a surjective morphism $\pi\colon S \to A$ satisfying that for every vertex $v \in \VA$, every outgoing edge $v \re c v' \in E(A)$, and every $s \in \pi^{-1}(v)$, there is an outgoing edge $s \re c s' \in E(S)$ with $\pi(s') = v'$. Recall that under our assumptions every vertex needs to have at least one outgoing edge, thus for every $v \in \VE$ and every $s \in \pi^{-1}(v)$ there must be at least one outgoing edge from $s$ in $S$.

The example arenas in this work are drawn following a~standard notation, where circles (resp.~squares) denote vertices controlled by Eve (resp.~Adam). Vertices with a single outgoing edge are denoted by a simple dot, it does not matter who controls them.

A strategy is \emph{positional} if $\pi$ is injective. In this case, we can assume that $V(S)=V(A)$ and $E(S)\subseteq E(A)$, with $\pi$ being identity.

An \emph{objective} is a set $W \subseteq C^\omega$ of infinite sequences of elements of $C$.
In this paper, we will always work with \emph{prefix\=/independent} objectives, meaning objectives which satisfy $cW=W$ for all $c \in C$ (that is $cw \in W \Leftrightarrow w \in W$ for any $w\in C^\omega$); this allows us to simplify many of the definitions.

A~graph~$G$ \emph{satisfies} an~objective~$W$ if $\lang(G) \subseteq W$.
A \emph{game} is given by a $C$\=/arena $A$ together with an~objective $W$.
It is \emph{winning} (for Eve) if there is a strategy $(S,\pi)$ such that $S$ satisfies~$W$.
In this case, we also say that Eve \emph{wins} the game $(A,W)$ with the strategy $(S,\pi$).
We say that an~objective~$W$ is \emph{positional} (over finite arenas or over arbitrary arenas) if for any (finite or arbitrary) arena~$A$, whenever Eve wins the game $(A,W)$, she wins $(A,W)$ with a~positional strategy.

\paragraph*{Neutral letters.}
A letter $\eps \in C$ is said to be \emph{weakly neutral} for an objective $W \subseteq C^\omega$ if for any word $w\in C^\omega$ decomposed into $w=w_0w_1\dots$ with non\=/empty words $w_i \in C^+$,
\[
    w \in W \iff \eps w_0 \eps w_1 \eps \dots \in W.
\]
A weakly neutral letter $\eps\in C$ is \emph{strongly neutral} if in the above, the $w_i$ can be chosen empty, and moreover, $\eps^\omega \in W$.
Thus, for prefix\=/independent objectives, the difference between the two notions relies on the membership of $\eps^\omega$ in $W$.

A~few examples: for the parity objective, the priority $0$ is strongly neutral; for $\Bounded$, the weight $0$ is strongly neutral; for $\MP_{\leq 0}$, the letter $0$ is only weakly neutral (because $1^\omega \notin \MP_{\leq 0}$ however $010010001\dots \in \MP_{\leq 0}$), and likewise for $\MP_{<0}$ because $0^\omega \notin \MP_{<0}$.

\paragraph*{Monotone and universal graphs.}
An \emph{ordered graph} is a graph $G$ equipped with a~total order $\geq$ on its set of vertices $V(G)$.
We say that it is \emph{monotone} if
\[
    v \geq u \re c u' \geq v' \tin G \qquad \text{implies} \qquad v \re c v' \in E(G).
\]
Such a graph is \emph{well founded} if the order ${\geq}$ on $V(G)$ is well founded.

We will use a variant of universality called (uniform) \emph{almost-universality} (for trees), which is convenient when working with prefix-independent objectives.
A $C$-graph $U$ is \emph{almost $W$\=/universal}, if $U$ satisfies $W$, and for any tree $T$ satisfying $W$, there is a vertex $t \in V(T)$ such that $T[t] \to U$.
We will rely on the following inductive result from~\cite{Ohlmann23}.

\begin{theorem}[Follows from Theorem~3.2 and Lemma~4.5 in~\cite{Ohlmann23}]\label{thm:structure_gives_positionality}
Let $W \subseteq C^\omega$ be a~prefix\=/independent objective such that there is a graph which is almost $W$\=/universal.
Then $W$ is positional over arbitrary arenas.
\end{theorem}

\paragraph*{Structuration results.} The following results were proved in Ohlmann's PhD thesis (Theorems~3.1 and~3.2 in~\cite{OhlmannThesis}); the two incomparable variants stem from two different techniques.

\begin{lemma}[Finite structuration]\label{lem:finite-structuration}
Let $W$ be a prefix-independent objective which is positional over finite arenas and admits a weakly neutral letter, and let $G$ be a finite graph satisfying~$W$.
Then there is a monotone graph $G'$ satisfying $W$ such that $G \to G'$.
\end{lemma}

\begin{lemma}[Infinite structuration]\label{lem:infinite-structuration}
Let $W$ be a prefix-independent objective which is positional over arbitrary arenas and admits a strongly neutral letter, and let $G$ be any graph satisfying~$W$.
Then there is a well-founded monotone graph $G'$ satisfying $W$ such that $G \to G'$.
\end{lemma}

Note that in both results, we may assume that $|G'| \leq |G|$, simply by restricting to the image of $G$.
Details of the proof of Lemma~\ref{lem:infinite-structuration} can be found in~\cite[Theorem~3]{Ohlmann23}; Lemma~\ref{lem:finite-structuration} appears only in Ohlmann's PhD thesis~\cite{OhlmannThesis}, we give details in Appendix~\ref{app:finite-structuration} for completeness.

\paragraph*{Automata.} A \emph{co\=/Büchi automaton over $C$} is a $q_0$\=/rooted $C\times\{\normal,\cobuchi\}$\=/graph $A$.
In this context, vertices $V(A)$ are called \emph{states}, edges $E(A)$ are called \emph{transitions}, and the root $q_0$ is called the \emph{initial state}.
Moreover, transitions of the form $q \re{(c,\normal)} q'$ are called \emph{normal transitions} and simply denoted $q \re c q'$, while transitions of the form $q \re{(c,\cobuchi)} q'$ are called \emph{co\=/Büchi} transitions and denoted $q \reb c q'$.
For simplicity, we assume automata to be \emph{complete} (for any state $q$ and any letter $c$, there is at least one outgoing transition labelled $c$ from $q$) and \emph{reachable} (for any state $q$ there is some path from $q_0$ to $q$ in $A$).

A path $q_0 \re{(c_0,a_0)} q_1 \re{(c_1,a_1)} \dots$ in $A$ is \emph{accepting} if it contains only finitely many co\=/B\"uchi transitions, meaning that only finitely many of $a_i$ equal $\cobuchi$.
If $q\in V(A)$ is a state, then define the \emph{language} $\lang(A,q) \subseteq C^\omega$ of a co-Büchi automaton \emph{from a~state} $q\in V(A)$ as the set of infinite words which label accepting paths from $q$ in $A$.
The \emph{language} of $A$ denoted $\lang(A)$ is $\lang(A,q_0)$.
Note that in this paper, automata are not assumed to be finite.

We say that an automaton is \emph{monotone} if it is monotone as a $C\times\{\normal,\cobuchi\}$\=/graph.
Likewise, morphisms between automata are just morphisms of the corresponding $C\times\{\normal,\cobuchi\}$\=/graphs that moreover preserve the initial state. Note that $A \to A'$ implies $\lang(A) \subseteq \lang(A')$.
A co\=/Büchi automaton is \emph{deterministic} if for each state $q\in V(A)$ and each letter $c \in C$ there is exactly one transition labelled by $c$ outgoing from $q$.

A \emph{resolver} for an automaton $A$ is a deterministic automaton $R$ with a morphism $R \to A$. Note that the existence of this morphism implies that $\lang(R)\subseteq \lang(A)$.
Such a resolver is \emph{sound} if additionally $\lang(R)\supseteq \lang(A)$ (and thus $\lang(R) = \lang(A)$).
A co\=/B\"uchi automaton is \emph{history\=/deterministic} if there exists a sound resolver $R$.
Our definition of history-determinism is slightly non-standard, but it fits well with our overall use of morphisms and of possibly infinite automata.
This point of view was also adopted by Colcombet (see~\cite[Definition~13]{Colcombet12}).
For more details on history determinism of co\=/Büchi automata, we refer to~\cite{HP06,KS15,BKS17,ARK22}.

We often make use of the following simple lemma, which follows directly from the definitions and the fact that composing morphisms results in a morphism.

\begin{lemma}\label{lem:morphisms_preserve_hd}
Let $A$, $A'$ be automata such that $A \to A'$, $A$ is history-deterministic, and $\lang(A) = \lang(A')$.
Then $A'$ is history-deterministic.
\end{lemma}

An automaton $A$ is \emph{saturated} if it has all possible co\=/Büchi transitions: $V(A) \times (C \times \{\cobuchi\}) \times V(A) \subseteq E(A)$.
The \emph{saturation} of an automaton $A$ is obtained from $A$ by adding all possible co-Büchi transitions. Similar techniques of saturating co\=/B\"uchi automata have been previously used to study their structure~\cite{KS15,IostiKuperberg2019,ARK22}.

Note that languages of saturated automata are always prefix\=/independent.
The lemma below states that co\=/Büchi transitions are somewhat irrelevant in history\=/deterministic automata recognising prefix\=/independent languages.

\begin{lemma}\label{lem:saturation}
Let $A$ be a history-deterministic automaton recognising a prefix-independent language and let $A'$ be its saturation.
Then $\lang(A)=\lang(A')$ and $A'$ is history-deterministic.
Moreover, $\lang(A') = \lang(A',q)$ for any $q \in V(A')$.
\end{lemma}

\begin{proof}
    Clearly $A \to A'$ thus $\lang(A) \subseteq \lang(A')$; it suffices to prove $\lang(A') \subseteq \lang(A)$ and conclude by Lemma~\ref{lem:morphisms_preserve_hd}.
    Let $w_0w_1\dots \in \lang(A')$ and let $q_0 \re{(w_0,a_0)} q_1 \re{(w_1,a_1)} \dots$ be an accepting path for $w$ in $A'$.
    Then for some $i$, $q_i \re{(w_i,a_i)} q_{i+1} \re{(w_{i+1},a_{i+1})} \dots$ is comprised only of normal transitions.
    Thus, this suffix of the path does not use edges added during the saturation process, which means this suffix is an~accepting path in $A$.
    We conclude that $w_iw_{i+1} \dots \in \lang(A, q_i)$. Due to the assumption that our automata are reachable, there is a finite word $u$ such that $q_0 \rp{u} q_i$ which implies that $uw\in\lang(A')$ but because of prefix independence it means that $w\in\lang(A')$.

    The claim that $\lang(A',q)$ is independent of $q$ follows directly from the fact that $A'$ is saturated: if $w_0w_1\dots\in\lang(A',q)$ as witnessed by a path starting with $q\re{(w_0,a_0)} q'$ and $q''$ is some other state then we can also accept $w$ from $q'$ by first using a co\=/Büchi transition $q''\re{(w_0,\cobuchi)}q'$ to reach $q'$ and then continue as the original path does.
\end{proof}

\section{Positional prefix-independent \texorpdfstring{$\bsigma{2}$}{Sigma02} objectives}
\label{sec:sigma2}

\subsection{A characterisation}

Recall that $\bsigma{2}$ objectives are countable unions of closed objectives; for the purpose of this paper it is convenient to observe that these are exactly those objectives recognised by (countable) deterministic co\=/Büchi automata (see for instance~\cite{mskrzypczak_colorings}).

The goal of the section is to prove Theorem~\ref{thm:charac-sigma2} which we now restate for convenience.

\charac*

Before moving on to the proof, we proceed with a quick technical statement that allows us to put automata in a slightly more convenient form.

\begin{lemma}\label{lem:little-massage}
Let $A$ be a history-deterministic automaton recognising a non\=/empty prefix\=/independent language.
There exists a history-deterministic automaton $A'$ with $\lang(A')=\lang(A)$ and such that from every state $q' \in V(A')$, there is an infinite path comprised only of normal transitions.
Moreover, if $A$ is countable, well founded, and monotone, then so is $A'$.
\end{lemma}

\begin{proof}
Let $V \subseteq V(A)$ be the set of states $q \in V(A)$ from which there is an infinite path of normal transitions.
Note that $V \neq \emptyset$ since $\lang(A)$ is non-empty. Due to prefix\=/independence of $A$ we can assume that $q_0$ belongs to $V$ (if not, modify the initial state using Lemma~\ref{lem:saturation}).

Since every path from $V(A) \setminus V$ visits at least one co-Büchi transition, we turn all normal transitions adjacent to states in $V(A) \setminus V$ into co-Büchi ones; moreover we saturate $A$. These operations do not affect $\lang(A)$ or history-determinism. Thus, without loss of generality assume that $A$ is already the modified automaton, with the initial state in $V$.

Call $A'$ the restriction of $A$ to states in $V$. It is clear that restricting $A$ to some subset of states preserves being countable, well founded, and monotone.

We claim that $\lang(A) = \lang(A')$.
The inclusion $\lang(A') \subseteq \lang(A)$ is obvious as $A'$ is a restriction of $A$ to a~subset of states. We focus on the converse: let $w=w_0w_1 \dots \in \lang(A)$ and take an accepting path~$\pi$ for $w$.
Then there is a suffix of $\pi$ which remains in $V$ and therefore defines a path in $A'$; we conclude thanks to prefix-independence of $\lang(A')$ (which is a~consequence of it being saturated).

It remains to see that $A'$ is history-deterministic.
For this, we observe that any transition adjacent to states in $V(A) \setminus V$ is a co-Büchi transition; therefore the map $\phi:V(A) \to V(A') = V$ which is identity on $V$ and sends $V(A) \setminus V$ to the initial state of $A'$ defines a~morphism $A \to A'$.
We conclude by Lemma~\ref{lem:morphisms_preserve_hd}.
\end{proof}

To prove Theorem~\ref{thm:charac-sigma2}, we separate both directions so as to provide more precise hypotheses.

\begin{lemma}\label{lem:positionality_to_structure}
Let $W$ be a prefix\=/independent $\bsigma{2}$ objective admitting a~strongly neutral letter. If $W$ is positional over arbitrary arenas then $W$ is recognised by a~countable history\=/deterministic monotone well\=/founded co\=/Büchi automaton.
\end{lemma}

\begin{proof}
Note that $W$ cannot be empty because it contains $\epsilon^\omega$.
Let $A$ be a history-deterministic co\=/Büchi automaton recognising $W$ with initial state $q_0$; thanks to Lemma~\ref{lem:little-massage} we assume that every state in $A$ participates in an~infinite path of normal transitions. Clearly $A$ is countable.
Let $G$ be the $C$\=/graph obtained from $A$ by removing all the co-Büchi transitions.
The fact that $G$ is sinkless (and therefore, $G$ is indeed a graph) follows from the assumption on $A$.
Since $W$ is prefix\=/independent, it holds that $G$ satisfies $W$.

Apply the infinite structuration result (Lemma~\ref{lem:infinite-structuration}, which requires the strongly neutral letter) to $G$ to obtain a well-founded monotone graph $G'$ satisfying $W$ and such that $G \re \phi G'$.
Note that we may restrict $V(G')$ to the image of~$\phi$. Due to the fact that $C$ is countable, this guarantees that $G'$ is countable.

Now let $A'$ be the co\=/B\"uchi automaton obtained from $G'$ by turning every edge into a~normal transition, setting the initial state to be $q'_0=\phi(q_0)$, and saturating.
Note that $A'$ is countable monotone and well-founded. Note moreover that the morphism $G \re \phi G'$ is in fact also a~morphism witnessing that $A\to A'$. Thus, since $A$ is history deterministic, Lemma~\ref{lem:morphisms_preserve_hd} implies that so is $A'$. It remains to see that $A'$ recognises $W$,.

Let $w \in \lang(A')$.
Then $w=uw'$ where $w' \in \lang(G') \subseteq W$.
It follows from prefix\=/independence that $w \in W$.
Conversely, let $w_0w_1\dots \in W$ as witnessed by an accepting path $\pi= q_0 \re{(w_0,a_0)} q_1 \re{(w_1,a_1)} \dots$ from $q_0$ in $A$. This path has only finitely many co-Büchi transitions.

Then consider the path $\pi'= \phi(q_0) \re{w_0} \phi(q_1) \re{w_1} \dots$ in $A'$, where we use co-Büchi transitions only when necessary, meaning when there is no normal transition $\phi(q_i) \re{w_i} \phi(q_{i+1})$ in $A'$.
Since $\pi$ visits only finitely many co-Büchi transitions, it is eventually a path in $G$, and thus since $\phi$ is a~morphism, $\pi'$ is eventually a path in $G'$, and hence it sees only finitely many co-Büchi transitions in $A'$. 
Hence $\lang(A')=W$.
\end{proof}

For the converse direction, we do not require a neutral letter.

\begin{lemma}\label{lem:automata_to_positionality}
If $W$ is a~prefix\=/independent objective recognised by a countable history\=/deterministic monotone well\=/founded co\=/B\"uchi automaton then
$W$ is positional over arbitrary arenas.
\end{lemma}

\begin{proof}
As previously, if $W$ is empty then it is trivially positional, so we assume that $W$ is non-empty, and we take an automaton $A$ satisfying the hypotheses above and apply Lemma~\ref{lem:little-massage} so that every state participates in an infinite path of normal transitions.
Let $U$ be the $C$-graph obtained from $A$ by removing all co-Büchi transitions and turning normal transitions into edges; thanks to Lemma~\ref{lem:little-massage}, $U$ is sinkless so it is indeed a graph.
We prove that $U$ is almost $W$-universal for trees.
Let $T$ be a tree satisfying $W$ and let $t_0$ be its root.

Since $A$ is history-deterministic, there is a mapping $\phi:V(T) \to V(A)$ such that for each edge $t \re c t' \in E(T)$, there is a transition $\phi(t) \re{(c,a)} \phi(t')$ in $A$ with some $a\in\{\normal,\cobuchi\}$, and such that for all infinite paths $t_0 \re{w_0} t_1 \re{w_1} \dots$ in $T$, there are only finitely many co-Büchi transitions on the path $\phi(t_0) \re{(w_0,a_0)} \phi(t_1) \re{(w_1,a_1)} \dots$ in $A$.

\begin{claim}
\label{cl:root}
There is a vertex $t'_0 \in V(T)$ such that for all infinite paths $t'_0 \re{w_0} t'_1 \re{w_1} \dots$ from $t'_0$ in $T$, there is no co-Büchi transition on the path $\phi(t'_0) \re{w_0} \phi(t'_1) \re{w_1} \dots$ in $A$.
\end{claim}

\begin{proof}[Proof of Claim~\ref{cl:root}]
Assume towards contradiction that no such vertex exists.
Then starting from the root $t_0$, we build an infinite path $t_0 \rp{w_0} t_1 \rp {w_1} \dots$ in $T$ such that $\phi(t_0) \rp{w_0} \phi(t_1) \rp{w_1} \dots$ has infinitely many co-Büchi transitions in $A$.
Indeed, assuming the path built up to $t_i$, we simply pick $t_i \rp{w_i} t_{i+1}$ such that there is a co-Büchi transition in $A$ on the corresponding path $\phi(t_i) \rp{w_i} \phi(t_{i+1})$. Thus, we constructed a path contradicting the observation below: this path has infinitely many co\=/B\"uchi transitions in $A$.
\end{proof}
There remains to observe that $\phi$ maps $T[t'_0]$ to $U$, and thus $U$ is almost $W$-universal for trees.
We conclude by applying Lemma~\ref{thm:structure_gives_positionality}.
\end{proof}

The above two lemmata constitute a proof of Theorem~\ref{thm:charac-sigma2}.

\subsection{A few examples}

We now elaborate on some examples.

\paragraph*{Kopczy\'nski-monotonic objectives.}
In our terminology, Kopczy\'nski's monotonic objectives correspond to the prefix-independent languages that are recognised by finite monotone co-Büchi automata.
Note that such automata are of course well-founded, but also they are history-deterministic (even determinisable by pruning): one should always follow a transition to a maximal state.
Therefore our result proves that such objectives are positional over arbitrary arenas.
A very easy example is the co-Büchi objective
\[
    \coBuchi=\{w \in \{\normal,\cobuchi\}^\omega \mid w \text{ has finitely many occurrences of } \cobuchi\},
\]
which is recognised by a (monotone) automaton with a single state.
Some more advanced examples are given in Figure~\ref{fig:k-monotone}.

\begin{figure}[ht]
\begin{center}
\includegraphics[width=0.9 \linewidth]{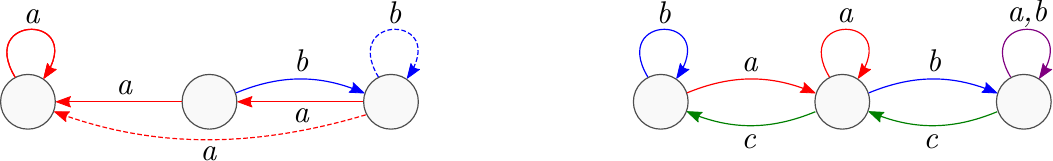}
\end{center}
\caption{Two finite monotone co-Büchi automata recognising prefix-independent languages.
For clarity, the co-Büchi transitions are not depicted but connect every pair of states; likewise, edges following from monotonicity (such as the dashed ones for example), are omitted.
The automaton on the left recognises words with finitely many $aab$ infixes.
The automaton on the right recognises words with finitely many infixes in $c(a^*cb^*)^+c$.}
\label{fig:k-monotone}
\end{figure}


\paragraph{Finite support.} The finite support objective is defined over $\N$ by
\[
    \Finite=\{w \in \N^\omega \mid \text{finitely many distinct letters appear in } w\}
\]


Consider the automaton $A$ over $V(A)=\N$ with 
\[
    v \re{w} v' \in E(A) \iff w, v' \leq v,  
\]
co-Büchi transitions everywhere, and initial state $0$ (see Figure~\ref{fig:finite}).
\begin{figure}[ht]
\begin{center}
\includegraphics[width=0.4\linewidth]{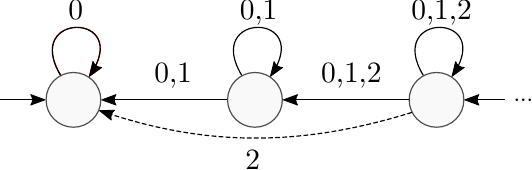}
\end{center}
\caption{An automaton $A$ for objective $\Finite$. Co-Büchi edges, as well as some edges following from monotonicity (such as the dashed one) are omitted for clarity.}\label{fig:finite}
\end{figure}

It is countable, history\=/deterministic, well\=/founded, and monotone and recognises $\lang(A) = \Finite$.
Details of the proof are easy and left to the reader.
Positionality of $\Finite$ can also be established by Corollary~\ref{cor:closure-union}, as it is a countable union of the safety languages $F^\omega \subseteq \N^\omega$, where $F$ ranges over finite subsets of $\N$.
As far as we are aware, this result is novel.\footnote{A similar positionality result is proved in~\cite{Gimbert2004}, but it assumes finite degree of the arena, vertex-labels (which is more restrictive), and injectivity of the colouring of the arena.}

\paragraph{Energy objectives.}
Recall the energy objective
\[
    \Bounded=\Big\{w_0 w_1 \dots \in \Z^\omega \mid \sup_{k}\sum_{i=0}^{k-1} w_i \text{ is finite}\Big\},
\]
which is prefix-independent and belongs to $\bsigma{2}$.
Consider the automaton $A$ whose set of states is $\N$, with the initial state $0$ and with all possible co\=/Büchi transitions, and normal transitions of the form $v \re{w} v'$ where $w \leq v - v'$.
Note that $A$ is well-founded and monotone, so we should prove that it is history-deterministic and recognises $\Bounded$.

Note that any infinite path of normal edges $v_0 \re{w_0} v_1 \re{w_1} \dots$ in $A$ is such that for all $i$, $w_i \leq v_i - v_{i+1}$, and therefore
\[
    \sum_{i=0}^{k-1} w_i \leq v_0 - v_k \leq v_0
\]
and thus $\lang(A) \subseteq \Bounded$.

A resolver for $A$ works as follows: keep a counter $c$ (initialised to zero), and from a~vertex $v$ and when reading an~edge $w$,
\begin{itemize}
\item if $v \geq w$, then take the normal transition $v \re{w} v-w$;
\item otherwise, take the co-Büchi transition $v \reb{w} c$ and increment the counter.
\end{itemize}

Formally, $R$ is defined by $V(R)=V(A) \times \N$, initial state $r_0=(0,0)$ and
\[
    \begin{array}{rcl}
    (v,c) \re{w} (v',c') \in E(R) &\iff& v'=v-w \geq 0 \tand c'=c \\
    (v,c) \reb{w} (v',c') \in E(R) &\iff& v-w<0 \tand c'=v'=c+1.
    \end{array}
\]
Clearly $(v,c) \mapsto v$ defines a morphism from $R$ to $A$ which sends $r_0$ to $q_0$, so there remains to see that $\Bounded \subseteq L(R)$.

Consider a word $w_0w_1\dots \in \Bounded$.
By definition, there exists $N$ such that
\[\tag{$*$}\label{eq:sums_bounded}
    \sup_{k}\sum_{i=0}^{k-1} w_i\leq N.
\]

Given a finite word $u \in \Z^*$, we let $\ss(u) \in \Z$ denote the sum of its letters.
Let $\pi$ be the unique path from $r_0$ in $R$ labelled by $w$.
Note that the counter (second coordinate) in states appearing in $\pi$ always grows, and that a co-Büchi transition is read precisely when it is incremented; we show that it cannot exceed $N$ which proves that the path is accepting.
Assume towards a contradiction that the counter exceeds $N$, therefore $\pi$ is of the form
\[
    (0,0) \rp{u_0} (v_0,0) \reb{u'_0} (1,1) \rp{u_1} (v_1,1) \reb{u'_1} \dots \reb{u'_{N-1}} (N,N) \rp{u_N} (v_N,N) \reb{u'_N} (N+1,N+1) \rp {w'}
\]
where the $u_i$'s are finite words, the $u'_i$'s are letters in $\Z$, they concatenate to $u_0u'_0 \dots u_N u'_N w' = w$, and for each $i \leq N$, it holds that $\ss(u_i)=i-v_i$ and $v_i - u'_i < 0$.
Therefore we have
\[
    \ss(u_0 u'_1 \dots u_N u'_N) = \underbrace{\ss(u_0)}_{=0-v_0} + \underbrace{u'_0}_{>v_0} + \dots + \underbrace{\ss(u_N)}_{=N-v_N} + \underbrace{u'_N}_{>v_N} \geq 0+1 + \dots + N > N,
\]
contradicting~\eqref{eq:sums_bounded}.
Hence $R$ is a sound resolver for $A$, so $A$ is history-deterministic.

We conclude that $\Bounded$ is positional over arbitrary arenas.

\paragraph{Eventually non-increasing objective.}
Over the alphabet $\N$, consider the objective
\[
    \ENI = \big\{w_0w_1 \dots \in \N^\omega \mid \text{there are finitely many } i \text{ such that } w_{i+1}>w_i\big\}.
\]
Note that since $\N$ is well founded, a~sequence belongs to $\ENI$ if and only if it is eventually constant.
Consider the automaton $A$ over $\N$ with the initial state $0$, with all possible co\=/Büchi transitions, and with normal transitions $v \re w v'$ if and only if $v \geq w \geq v'$.
Note that $A$ is countable, well\=/founded, and monotone, so we should prove that it recognises $\ENI$ and is history\=/deterministic.

First, note that any infinite path of normal edges $v_0 \re{w_0} v_1 \re{w_1} \dots$ in $A$ is such that $v_0 \geq w_0 \geq v_1 \geq w_1 \geq \dots$, and therefore $\lang(A) \subseteq \ENI$.
A sound resolver for $A$ simply goes to the state $w$ when reading a~letter $w$, using a normal transition if possible, and a co-Büchi transition otherwise.
We leave the formal definition to the reader.

\paragraph{Eventually non-decreasing objective.}
In contrast, the objective
\[
    \END = \{w_0w_1 \dots \in \N^\omega \mid \text{there are finitely many } i \text{ such that } w_{i+1}<w_i\}
\]
is not positional over arbitrary arenas, as witnessed by Figure~\ref{fig:counter-example-arena}.

\begin{figure}[ht]
\begin{center}
\includegraphics[width=0.55\linewidth]{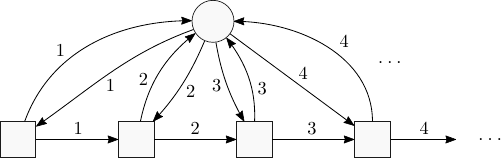}
\end{center}
\caption{An arena over which Eve requires a non-positional strategy in order to produce a~sequence which is eventually non-decreasing.}\label{fig:counter-example-arena}
\end{figure}

\subsection{Closure under countable unions}

We now move on to Corollary~\ref{cor:closure-union}, which answers Kopczy\'nski's conjecture in the affirmative in the case of $\bsigma{2}$ objectives.

\closureunion*

\begin{proof}
Let $W_0,W_1,\dots$ be a family of countably many prefix-independent $\bsigma{2}$ objectives admitting strongly neutral letters.
Using Theorem~\ref{thm:charac-sigma2} we get countable history-deterministic well-founded monotone co\=/B\"uchi automata $A_0,A_1,\dots$ for the respective objectives; without loss of generality we assume that they are saturated (Lemma~\ref{lem:saturation}).

Then consider the automaton $A$ obtained from the disjoint union of the $A_i$'s by adding all possible co-Büchi transitions, and all normal transitions from $A_i$ to $A_j$ with $i>j$.
The initial state in $A$ can be chosen arbitrarily.
Note that $A$ is well-founded, monotone, and countable, so we should prove that it recognises $W=\bigcup_i W_i$ and is history deterministic.

Note that any infinite path in $A$ which visits finitely many co-Büchi transitions eventually remains in some $A_i$, and thus by prefix-independence, $\lang(A) \subseteq W$.

It remains to prove history\=/determinism of $A$. Let $R_0,R_1,\dots$ be resolvers for $A_0,A_1,\dots$ witnessing that these automata are history deterministic. Consider a resolver which stores a sequence of states $(r_0,r_1,\ldots)$, with $r_i$ being a state of $R_i$.
Initially these are all initial states of the respective resolvers and the transitions follow the transitions of all the resolvers synchronously. 
Additionally, we store a round\=/robin counter, which indicates one of the resolvers, following the sequence
$R_0;R_0,R_1;R_0,R_1,R_2;R_0,R_1,R_2,R_3;\dots$
If we see a normal transition in the currently indicated resolver, then we also see a normal transition in $R$, and otherwise, we update the counter to the next resolver and see a co\=/Büchi transition in $R$.

Let $\ell_0 \ell_1 \dots = 0010120123 \dots$.
We formally define the resolver $R$ as follows: $V(R)=\Pi_{i=0}^\infty V(R_i) \times \N$, with transitions
\[
    \begin{array}{rcl}
        (r,j) \re{(c,\normal)} (r',j') &\iff& [\forall i. \exists a_i.\ r_i \re{(c,a_i)} r'_i \in E(R_i)]\\
&& \tand r_{\ell_j} \re{(c,\normal)} r'_{\ell_j} \in E(R_{\ell_j}) \tand j'=j \\
        (r,j) \re{(c,\cobuchi)} (r',j') &\iff & [\forall i. \exists a_i.\ r_i \re{(c,a_i)} r'_i \in E(R_i)]\\
&& \tand r_{\ell_j} \re{(c,\cobuchi)} r'_{\ell_j} \in E(R_{\ell_j}) \tand j'=j{+}1
    \end{array}
\]
and morphism $\phi:(r,j) \mapsto \phi_j(r_j) \in V(A)$, where $\phi_0:R_0 \to A_0,\phi_1:R_1 \to A_1,\dots$ are the respective morphisms.

We now prove that $W\subseteq\lang(A)$ and the above resolver is sound.
For that, consider a word $w$ which belongs to $\lang(A_n)$ for some $n$.
Assume for the sake of contradiction that the path in~$A$ constructed by the above resolver reading $w$ contains infinitely many co\=/B\"uchi transitions.
It means that infinitely many times the resolver $R_n$ reached a co\=/B\"uchi transition in $A_n$.
But this contradicts the assumption that $R_n$ is sound.
We conclude that $W$ is positional by applying Lemma~\ref{lem:automata_to_positionality}.
\end{proof}

\section{From finite to arbitrary arenas}\label{sec:finite-to-infinite}

In this section we study the difference between positionality over finite and arbitrary arenas.

\subsection{Mean-payoff games}


There are, in fact, four non\=/isomorphic variants of the mean\=/payoff objective. Three of them fail to be positional over arbitrary arenas (even over bounded degree arenas), as expressed by the following facts.

\begin{proposition}
\label{pro:non-positional-limsup}
The mean\=/payoff objective $\MP_{\leq 0}$ over $w_0w_1\dots\in \Z^\omega$ with the condition $\limsup_k \frac 1 k \sum_{i=0}^{k-1} w_i \leq 0$ is not positional over arbitrary arenas.
\end{proposition}

\begin{figure}[ht]
\begin{center}
\includegraphics[width=0.55\linewidth]{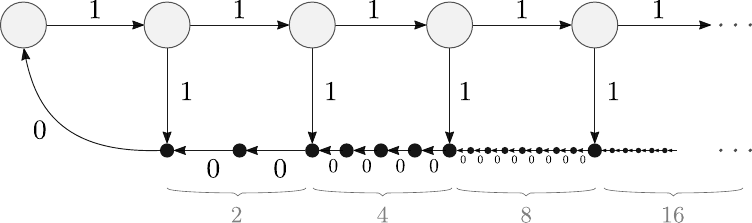}
\end{center}
\caption{The arena used in the proof of \cref{pro:non-positional-limsup}.}
\label{fig:counter-example-limsup}
\end{figure}

\begin{proof}
Consider the arena depicted on \cref{fig:counter-example-limsup}. Eve can win by following bigger and bigger loops which reach arbitrarily far to the right. This strategy brings the average of the weights closer and closer to $0$.

Nevertheless, each positional strategy of Eve either moves infinitely far to the right (resulting in $\lim_k \frac 1 k \sum_{i=0}^{k-1} w_i=1$) or repeats some finite loop which results in a fixed positive limit $\lim_k \frac 1 k \sum_{i=0}^{k-1} w_i>0$. In both cases it violates $\MP_{\leq 0}$.
\end{proof}

\begin{proposition}
\label{pro:non-positional-liminf}
Consider two $\liminf$ variants of the mean\=/payoff objective over $w_0w_1\dots\in \Z^\omega$: one where we require that $\liminf_k \frac 1 k \sum_{i=0}^{k-1} w_i \leq 0$, and the other where that same quantity is $<0$. Both these objectives are not positional over arbitrary arenas.
\end{proposition}

\begin{figure}[ht]
\begin{center}
\includegraphics[width=0.65\linewidth]{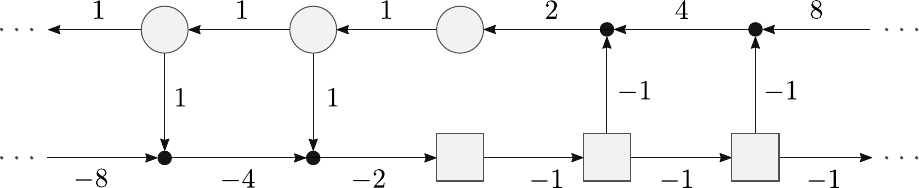}
\end{center}
\caption{The arena used in the proof of \cref{pro:non-positional-liminf}.}
\label{fig:counter-example-liminf}
\end{figure}

\begin{proof}
Consider the arena depicted on \cref{fig:counter-example-liminf}. Again, Eve has a winning strategy for both these objectives by always going sufficiently far to the left, to ensure that the average drops below for instance $-\frac 1 2$.

Nevertheless, each positional strategy of Eve either moves infinitely far to the left (resulting again in $\lim_k \frac 1 k \sum_{i=0}^{k-1} w_i=1$), or repeats some finite loop, reaching a minimal negative weight $-2^n$ for some $n>0$. Now, Adam can win against this strategy by repeating a loop going to the right, in such a way to reach a weight $2^{n+1}$. The label of such a path satisfies $\lim_k \frac 1 k \sum_{i=0}^{k-1} w_i=\frac{2^{n+1}-1}{4n+4}>0$, violating both objectives.
\end{proof}

The remaining fourth type of a~mean\=/payoff objective is ,,$\limsup<0$'':
\[
    \MP_{<0} = \Big\{w_0 w_1 \dots \in \Z^\omega \mid \limsup_k \frac 1 k \sum_{i=0}^{k-1} w_i < 0\Big\}.
\]

\begin{proposition}
The objective $\MP_{<0}$ is positional over arbitrary arenas.
\end{proposition}

\begin{proof}
Consider the tilted boundedness objective with parameter $n \geq 1$, defined as
\[
    \TB n = \Big\{w_0w_1 \dots \in \Z^\omega \mid \sup_k \sum_{i=0}^{k-1}(w_i+\frac{1}{n}) \text{ is finite}\Big\}
\]
Note that renaming weights by $w \mapsto (1{+}n w)$ reduces $\TB n$ to $\Bounded$. More precisely, for every $w_0w_1\dots \in \Z^\omega$ we have
\[
w_0w_1\dots \in \TB n \Longleftrightarrow (1{+}nw_0)(1{+}nw_1)\dots \in \Bounded,
\]
because the bound is just multiplied by the constant $n$.

Therefore, since $\Bounded$ is positional over arbitrary arenas, it follows that $\TB n$ is also positional over arbitrary arenas: rename the labels, find a~positional strategy and then recall original labels. Note also that for every $n$ the objective $\TB n$ belongs to $\bsigma{2}$, as a~union ranging over $N\in\N$ of closed (in other words safety) objectives $\big\{w_0w_1\dots\in\Z^\omega\mid \forall_{k\in\N} \sum_{i=0}^{k-1}(w_i+\frac{1}{n})\leq N\big\}$.

\begin{claim}\label{claim:mp_is_union_of_tbs}
It holds that $\MP_{<0} = \bigcup_{n \geq 1} \TB n$.
\end{claim}

\begin{proof}[Proof of Claim~\ref{claim:mp_is_union_of_tbs}]
Write $\mp(w)=\limsup_k \frac{1}{k} \sum_{i=0}^{k-1} w_i$.
If
\[
w=w_0w_1\dots \in \TB n
\] then there is a bound $N$ such that for all $k$, $\sum_{i=0}^{k-1} (w_i + \frac{1}{n}) \leq N$, therefore $\frac{1}{k} \sum_{i=0}^{k-1} w_i \leq \frac{N}{k} - \frac{1}{n}$ and thus $\mp(w) \leq -\frac{1}{n} <0$, so $w \in \MP_{<0}$.
Conversely, if $w \in \MP_{<0}$ and $n$ is large enough so that $\frac{1}{n} \leq \mp(w)$, then $w \in \TB n$.
\end{proof}

Now, positionality of $\MP_{<0}$ follows from the claim together with Corollary~\ref{cor:closure-union}, as all $\TB n$ are prefix\=/independent, admit a~strongly neutral letter, are positional, and belong to~$\bsigma{2}$.\footnote{We thank Lorenzo Clemente for suggesting to use closure under union. A~direct proof (constructing a~universal graph) is available in the unpublished preprint~\cite{OhlmannMP}.}
\end{proof}

\subsection{A completeness result}

The aim of this section is to prove Theorem~\ref{thm:main} below, which roughly states that any objective which is positional over finite arenas is equivalent to an~objective which is positional over arbitrary arenas.
First, we need to define what we mean by ``equivalent''.

\paragraph*{Equivalence over finite arenas}

Recall that two prefix-independent objectives $W,W' \subseteq C^\omega$ are said to be \emph{finitely equivalent}, written $W \equiv W'$, if for all finite $C$-arenas $A$,
\[
    \text{ Eve wins } (A,W) \quad \iff \quad \text{ Eve wins } (A,W').
\]
Since one may view strategies as games controlled by Adam, we obtain the following motivating result.

\begin{lemma}
If $W \equiv W'$ and $W$ is positional over finite arenas then so is $W'$.
\end{lemma}

\begin{proof}
Let $A$ be a finite $C$-arena such that Eve wins $(A,W')$.
Then Eve wins $(A,W)$, so she wins with a positional strategy $S$.
Looking at $S$ as a finite $C$-arena controlled by Adam yields that Eve wins $(S,W')$, thus $S$ satisfies $W'$.
\end{proof}

We now move on to the proof of our completeness result.

\main*

We start with the following observation, which is a standard topological argument based on K\"onig's lemma. Note that the assumption of finiteness of~$G$ is essential here.

\begin{restatable}[]{lemma}{closed}
\label{lem:finite_graph_in_sigma2}
Let $G$ be a finite $C$-graph and $v \in G$.
Then $\lang(G,v)$ is a closed subset of $C^\omega$.
\end{restatable}

\begin{proof}
    Consider the language $K \subseteq C^*$ of finite words $w$ such that there is no path from $v$ labelled by $w$.
    By definition, $K C^\omega$ is open.
    Our goal is to show that its complement coincides with $\lang(G,v)$.
    Clearly $\lang(G,v) \cap K\cdot C^\omega=\emptyset$. For the remaining inclusion assume that $w\in C^\omega$ is an infinite word such that $w\notin K\cdot C^\omega$. We want to show that $w\in \lang(G,v)$. The assumption that $w\notin K\cdot C^\omega$ means that every prefix of $w$ labels a finite path from $v$ in~$G$. By applying K\"onig's lemma, relying on the fact that $E(G)$ is finite, we see that $w$ must label an infinite path in $G$ from $v$, thus $w\in \lang(G,v)$.
\end{proof}

We may now give the crucial definition.
Given a prefix-independent objective $W \subseteq C^\omega$, we define its finitary substitute to be
\[
    \fin W = \{w \in C^\omega \mid \text{$w$ labels a path in some finite graph $G$ which satisfies $W$}\}.
\]
Note that $\fin W \subseteq W$ and $\fin W$ is prefix\=/independent (we can prepend an additional finite path to any graph).
Now observe that
\[
    \fin W = \bigcup_{\substack{G \text{ finite graph} \\ \text{$G$ satisfies $W$}}} \lang(G) = \bigcup_{\substack{G \text{ finite graph} \\ G \text{ satisfies } W \\ v \in V(G)}} \lang(G,v),
\]
and since there are (up to isomorphism) only countably many finite graphs, it follows from Lemma~\ref{lem:finite_graph_in_sigma2} that $\fin W \in \bsigma{2}$.

\begin{lemma}\label{lem:substitute_is_equivalent}
    Let $W \subseteq C^\omega$ be a prefix-independent objective which is positional over finite arenas. Then $\fin W \equiv W$.
\end{lemma}
    
\begin{proof}
    Let $A$ be a finite $C$-arena.
    Since $\fin W \subseteq W$, it is clear that if Eve wins $(A,\fin W)$, then she wins $(A,W)$.
    Conversely, assume Eve wins $(A,W)$.
    Then she has a positional strategy~$S$ in $A$ which is winning for $W$.
    Since $S$ is a finite graph, it is also winning for $\fin W$ and therefore Eve wins $(A,\fin W)$.
\end{proof}

We should make the following sanity check.

\begin{lemma}\label{lem:wfin-pi}
If $W$ is prefix-independent, then $\fin W$ as well.
\end{lemma}

\begin{proof}
Take a letter $c \in C$, we aim to show that $c\fin W = \fin W$.
Let $w \in c\fin W$, and let $G$ be a finite graph satisfying $W$ such that $cw$ labels a path from $v \in V(G)$ in $G$.
Then $w$ labels a~path from a $c$-successor of $v$ in $G$, thus $w \in \fin W$.

Conversely, let $w \in \fin W$, and let $G$ be a finite graph satisfying $W$ such that $w$ labels a~path from $v \in V(G)$ in $G$.
Let $G'$ be the graph obtained from $G$ by adding a fresh vertex~$v'$ with a unique outgoing $c$-edge towards $v$.
Since $W$ is prefix-independent, $G'$ satisfies $W$.
Since $cw$ labels a path from $v'$ in $G'$, it follows that $cw \in \fin W$.
\end{proof}


We are now ready to prove Theorem~\ref{thm:main}.

\begin{proof}[Proof of Theorem~\ref{thm:main}]
Let $W$ be a prefix-independent objective which is positional over finite arenas and admits a weakly neutral letter $\eps$.
We show that $\fin W$ satisfies the assumptions of Lemma~\ref{lem:automata_to_positionality} which guarantees that it is positional over arbitrary arenas. Since \cref{lem:substitute_is_equivalent} implies that $\fin W\equiv W$, this concludes the proof of Theorem~\ref{thm:main}.

Thanks to Lemma~\ref{lem:finite-structuration}, any finite graph $H$ satisfying $W$ can be embedded into a monotone finite graph $G$ which also satisfies $W$; note that $\lang(H) \subseteq \lang(G)$.
Therefore
\[
    \fin W = \bigcup_{\substack{H \text{ finite graph} \\ H\text{ satisfies W}}} \lang(H) = \bigcup_{\substack{G \text{ finite monotone graph} \\ G\text{ satisfies W}}} \lang(G).
\]
Let $G_0,G_1,\dots$ be an enumeration (up to isomorphism) of all finite monotone graphs satisfying~$W$.
Then consider the automaton $A$ obtained from the disjoint union of the $G_i$'s by adding all normal transitions from $G_i$ to $G_j$ for $i > j$, and saturating with co-Büchi transitions.
The initial state $q_0$ is chosen to be $\max V(G_0)$, the maximal state in $G_0$.
Note that $A$ is countable, monotone, and well founded, so there remains to prove that $\lang(A)=\fin W$ and that $A$ is history\=/deterministic.

Clearly for any monotone graph $G$ satisfying $W$, it holds that $\lang(G) \subseteq \lang(A)$, and thus $\fin W \subseteq \lang(A)$.
Conversely, let $w \in \lang(A)$, and consider an accepting path $\pi$ for $W$.
Then eventually, $\pi$ visits only normal edges, and therefore eventually, $\pi$ remains in some $G_i$.
Thus $w=uw'$ with $w' \in \lang(G_i) \subseteq \fin W$, we conclude by prefix-independence of $\fin W$ (Lemma~\ref{lem:wfin-pi}).

To prove that $A$ is history deterministic we now build a resolver: intuitively, we deterministically try to read in $G_0$, then if we fail, go to $G_1$, then $G_2$ and so on.
The fact that reading in each $G_i$ can be done deterministically follows from monotonicity: for each $v \in V(G_i)$ and each $c \in C$, the set $\{v' \in V(G_i) \mid v \re c v' \in E(G_i)\}$ of $c$-successors of $v$ is downward closed.
We let $\delta_i(v,c)$ denote the maximal $c$-successor of $v$ in $G_i$ if it it exists, and $\delta_i(v,c)=\bot$ if $v$ does not have a $c$-successor.
It is easy to see that in a monotone graph $G$, $v \leq v'$ implies $\lang(G,v) \subseteq \lang(G,v')$; in words, more continuations are available from bigger states.

Now we define the resolver $R$ by $V(R)=V(A)$, $r_0=q_0=\max V(G_0)$, and for any $q,q' \in V(A)$ and $c\in C$,
\[
    \begin{array}{rcl}
    q \re c q' \in E(R) &\iff& \exists i, q,q' \in V(G_i) \tand q'=\delta_i(q) \neq \bot \\
    q \reb c q' \in E(R) & \iff & \exists i, q \in V(G_i) \tand \delta_i(q,c) = \bot \tand q' = \max V(G_{i+1}).
    \end{array}
\]
Clearly, $R$ is deterministic and $R \to A$ so $R$ is a~resolver; it remains to prove soundness.
Take $w \in \lang(A)$ and let $i$ such that $w \in \lang(G_i)$ (from some moment on we need to stop using transitions which change components and $\lang(A)=\fin W$).
Let $\pi$ be the unique path from $r_0=\max V(G_0)$ in $R$ labelled by $w$.
We claim that $\pi$ remains in $\bigcup_{j \leq i}V(G_j)$ and thus it can only visit at most $i$ co-Büchi transitions, so it is accepting.
Assume for contradiction that $\pi$ reaches $V(G_{i+1})$.

Then it is of the form $\pi=\pi_0 \pi_1 \dots \pi_i \pi'$ where each $\pi_j$ is a path from $\max(V(G_j))$ in~$G_j$ and $\pi'$ starts from $\max(G_{i+1})$.
Let $w_0, w_1, \dots, w_i$ and $w'$ be the words labelling the paths, so that $w=w_0w_1 \dots w_i w'$.
Denote $q=\max(V(G_i))$.
Then $w_i$ is not a label of a~finite path from $q$ in $G_i$, therefore $w_i w' \notin \lang(G_i,q)=\lang(G_i)$.
At the same time $w \in \lang(G_i)$ thus $q \rp{w_0 \dots w_{i-1}} q' \rp{w_i w'}$ for some $q' \in V(G_i)$.
But then $w_i w' \in L(G_i,q') \subseteq L(G_i,q)$, a~contradiction.
\end{proof}

\section{Conclusion}

We gave a characterisation of prefix-independent $\bsigma{2}$ objectives which are positional over arbitrary arenas as being those recognised by countable history-deterministic well-founded monotone co-Büchi automata. We moreover deduced that this class is closed by unions.
We proved that, with a proper definition, mean-payoff games are positional over arbitrary arenas.
Finally, we showed that any prefix-independent objective which is positional over finite arenas is finitely equivalent to an objective which is positional over arbitrary arenas.

\paragraph*{Open questions.} There are many open questions on positionality.
Regarding $\bsigma{2}$ objectives, the remaining step would be to lift the prefix-independence assumptions; this requires some new techniques as the proofs presented here do not immediately adapt to this case.

As mentioned in the introduction, Casares~\cite{CasaresThesis} obtained a characterisation of positional $\omega$\=/regular objectives, while we characterised (prefix-independent) $\bsigma{2}$ positional objectives.
A~common generalisation, which we see as a far reaching open question would be to characterise positionality within $\bdelta{3}$; hopefully establishing closure under union for this class.

Another interesting direction would be to understand finite memory for prefix\=/independent $\bsigma{2}$ objectives; useful tools (such as structuration results) are already available~\cite{CO25}.
A related (but independent) path is to develop a better understanding of (non-prefix-independent) closed objectives, which so far has remained elusive.

\clearpage

\bibliographystyle{alphaurl}
\bibliography{bib}

\newpage

\appendix

\section{Proof of the finite structuration result}\label{app:finite-structuration}

We include a proof of Lemma~\ref{lem:finite-structuration}; it is identical to the one in~\cite{OhlmannThesis}; a similar result (with the same proof) also appears in~\cite[Theorem~4.8]{CFGO22}.

\begin{proof}[Proof of Lemma~\ref{lem:finite-structuration}]
    Let $W \subseteq C^\omega$ be a prefix-independent objective admitting a weakly neutral letter $\eps$ and which is positional over finite arenas, and let $G$ be a finite graph satisfying $W$.
    We let $G^\eps$ be a graph obtained from $G$ by saturating it with $\eps$-edges: we successively add arbitrary $\eps$-edges until obtaining a graph satisfying $W$ but such that adding any $\eps$-edge would create a path whose label does not belong to $W$.
    Note that $G \to G^\eps$.
    We claim that the relation $>$ defined by 
    \[
        v > v' \qquad \iff \qquad v \neq v' \tand v \re \eps v' \in E(G^\eps)
    \]
    defines a strict total pre-order over $V(G)$.
    Transitivity is easy to prove: if $v \re \eps v' \re \eps v''$ in $G^\eps$ then adding the edge $v \re \eps v''$ cannot create a path whose label does not belong to $W$.
    
    The difficulty lies in establishing totality.
    Let $v_0 \neq v_1$ be such that neither $v_0 \re \eps v_1$ nor $v_1 \re \eps v_0$ belong to $E(G^\eps)$.
    Then consider the arena $A$ defined by $\VA=V(G)=V(G^\eps)$, $\VE=\{\bullet\}$, where $\bullet \notin V(G)$, and
    \[
        E(A) = E(G^\eps) \cup \{v \re c \bullet \mid \exists p \in \{0,1\}, v \re c v_p \in E(G^\eps)\} \cup \{\bullet \re \eps v_p \mid s \in \{0,p\}\}.
    \]
    Observe that Eve wins the game $(A,W)$, simply by applying the strategy that heads to $v_s$ after entering $\bullet$ via an edge $v \re c \bullet$ such that $v \re c v_s \in E(G^\eps)$ (we abstain from giving a~more formal definition).
    Therefore, Eve wins with a positional strategy $(S,\pi)$; without loss of generality we assume that it chooses $v_0$, meaning that $v_0$ is an $\eps$-successor of $\pi^{-1}(\bullet)$ in $S$.
    
    Call $G'$ the graph obtained from adding the edge $v_1 \re \eps v_0$ to $G^\eps$, we claim that $G'$ satisfies~$W$, which contradicts the fact that $G^\eps$ is saturated with $\eps$-edges.
    Indeed, for any path in~$G'$, we observe that there is a path in $S$ with the same labels: it is enough to replace each subpath of the form $v\re c v_1 \re \eps v_0$ by $v \re c \bullet \re \eps v_0$ (and possibly remove the first edge if it is labelled by $\eps$).
    Since $S$ satisfies $W$ (it is a~winning strategy), and $\eps$ is neutral, we conclude that $G'$ satisfies $W$ and thus the edge $v_1 \re \eps v_0$ should belong to $E(G^\eps)$. Therefore, $>$ is a~strict pre-order.

    We now let $\bar G$ be the graph defined over $V(\bar G)=V(G)=V(G^\eps)$ by 
    \[
        E(\bar G) = \{v \re c v' \mid \exists u,u', v \rp {\eps^*} u \re c u' \rp {\eps^*} v' \tin G^\eps\}.
    \]
    It is a direct check that $\bar G$ satisfies $W$ (by neutrality of $\eps$) and that it is monotone with respect to the strict pre-order $>$ (by definition).
    Note that $G^\eps \to \bar G$.

    Now we would like to turn $>$ into a~linear order by combining it's equivalence classes. For that, observe that vertices $v \neq v'$ such that $v > v'$ and $v' > v$, have identical incoming and outgoing edges in $\bar G$.
    Therefore, the graph $G'$ defined over $V(G') = V(G) / \sim$ where $v \sim v' \iff v>v'>v$ by
    \[
        E(G') = \{[v] \re c [v'] \mid v \re c v' \tin E(\bar G)\},
    \]
    makes sense, $v \mapsto [v]$ defines a~morphism $\bar G \to G'$, and $>$ induces a total strict order over $V(G')$ which makes it monotone.
    This concludes the proof.
\end{proof}

\end{document}